\documentclass[11pt]{article}

\usepackage[tbtags]{amsmath}
\usepackage{amssymb}
\usepackage{amsthm}
\usepackage[misc]{ifsym}
\usepackage{cases}
\usepackage{cite}
\usepackage{mathrsfs}
\usepackage{xcolor}
\usepackage{graphicx}
\usepackage{float}
\usepackage{color}
\usepackage{hyperref}
\usepackage{amsmath}

\hypersetup{hypertex=true,
	colorlinks=true,
	linkcolor=red,
	anchorcolor=blue,
	citecolor=blue}

\newcommand{\argmin}{\mathop{\mathrm{argmin}}}
\numberwithin{equation}{section}
\normalsize

\title{\bf Two Stochastic Control Methods for Mean-Variance Portfolio Selection of Jump Diffusions and Their Relationship
	\thanks{This work is financially supported by the National Key R\&D Program of China (2022YFA1006104), National Natural Science Foundations of China (12471419, 12271304), and Shandong Provincial Natural Science Foundations (ZR2024ZD35, ZR2022JQ01).}}
\author{\normalsize
	Qiyue Zhang\thanks{\it School of Mathematics, Shandong University, Jinan 250100, P.R. China, E-mail: qiyuezhang@mail.sdu.edu.cn},\quad
    Jingtao Shi\thanks{\it Corresponding author, School of Mathematics, Shandong University, Jinan 250100, P.R. China, E-mail: shijingtao@sdu.edu.cn}}
\date{}

\newtheorem{mypro}{Proposition}[section]
\newtheorem{mythm}{Theorem}[section]
\newtheorem{mydef}{Definition}[section]
\newtheorem{mylem}{Lemma}[section]

\begin{document}
	
\maketitle

\noindent{\bf Abstract:}\quad
	This paper is concerned with the maximum principle and dynamic programming principle for mean-variance portfolio selection of jump diffusions and their relationship. First, the optimal portfolio and efficient frontier of the problem are obtained using both methods. Furthermore, the relationship between these two methods is investigated. Specially, the connections between the adjoint processes and value function are given.
	
	\vspace{2mm}
	
	
	\vspace{2mm}
	
	
\section{Introduction}

Portfolio selection is to maximizing return and minimizing risk by constructing a portfolio of risk-free bonds and risky assets. Mean-variance model was proposed by Nobel Prize-winning economist Markowitz in 1952. After his pioneering work, many scholars contribute to the study of mean-variance problem. Mean-variance problem in fact is a two-objective stochastic optimal control problem. Among others, there are two common approaches to overcome the difficulty of multiple objectives. One method is to embed this problem into a stochastic {\it linear-quadratic} (LQ) problem, proposed by Zhou and Li \cite{ZL00}. Another one is to consider this problem as a stochastic optimal control problem of mean field type, proposed by Anderson and Djehiche \cite{AD11}.

For mean-variance portfolio selection of jump diffusions, Guo and Xu \cite{GX04,GX07} solve this problem by HJB equation and verification theorem; Framstad et al. \cite{FOS04} solve it by the {\it maximum principle} (MP); Shen and Siu \cite{SS13} solve it by the MP for a mean field model.

There are two common methods for solving stochastic control problems: Pontryagin's {\it maximum principle} and Bellman's {\it dynamic programming principle} (DPP). Many scholars have researched the relationship between them. Zhou \cite{Zhou90} researched the relationship between MP and DPP in stochastic control. Shi and Wu \cite{SW11} researched relationship between MP and DPP for stochastic control of jump diffusions.

Framstad et al. \cite{FOS04} used MP to study mean-variance portfolio selection of jump diffusions. However, they imposed a restriction on the mean value of the investor's wealth process, which does not reflect actual market conditions. This condition can be removed from this paper. In Guo and Xu \cite{GX04,GX07}, the price process of the stock is evolved as a compensated Poisson process. In a more general framework, the compensated Poisson process can be naturally generalized to a Poisson random measure. To the best of our knowledge, there are no results about the relationship between MP and DPP for mean-variance portfolio selection of jump diffusions. In this paper, we will research this topic.

The contribution of this paper can be summarized as follows.

(1) By embedding the problem into a stochastic optimal control problem, we can solve the mean-variance portfolio selection of jump diffusions directly, and find out the optimal control and efficient frontier. So the constraint on the mean value in \cite{FOS04} can be removed. Furthermore, we have derived the efficient frontier of the problem.

(2) In \cite{GX04}, there is a misusing of the generalized It\^{o}'s formula (see proof of Theorem 1 of their paper). We present the correct result in Lemma 4.1 of this paper. Moreover, we establish the DPP with Poisson random measure in the state equation, then we use the DPP to solve our problem and find out the optimal control.

(3) After solving the problem with MP and DPP respectively, we examine the relationship between them for mean-variance portfolio selection of jump diffusions, which has no previous research on this subject.

The remaining sections of this paper are organized as follows. In section 2, we state our problem and embed our original problem into a stochastic LQ problem of jump diffusions. In section 3, the stochastic LQ problem of jump diffusions is solved by the MP. The optimal control and efficient frontier of the problem are obtained. In section 4, the DPP is established for systems governed by {\it stochastic differential equations} (SDEs) with Poisson random measures. Building on this result, we proceed to characterize the optimal control of the problem. In section 5, we obtain the relationship between MP and DPP for the problem.

\section{Problem statement}

Let $T > 0$ be a finite time duration and let $(\Omega ,\mathcal{F},\mathbb{P})$ be a complete probability space, equipped with a one-dimensional standard Brownian motion $\left \{ B(t) \right \} _{0\le t \le T}$ and a Poisson random measure $N(\cdot , \cdot)$ independent of $B(\cdot)$ with the intensity measure $\hat{N} (dt,dz)=\lambda(dz)dt$. We write $\tilde{N} (dt,dz):= N(dt,dz)-\lambda(dz)dt$ for the compensated Poisson martingale measure.

Consider a market consisting of a risk-free bond and a risky asset. The price process $S_0(t)$ of the risk-free bond at time $t\in [0,T]$ is given by
\begin{equation*}
dS_0 (t)= \rho_tS_0 (t)dt,
\end{equation*}
where $\rho_\cdot> 0$, is a bounded, deterministic continuous function on $[0,T]$.
The price process $S_1 (t)$ of the risky asset at time $t\in [0,T]$ satisfies the following SDE of jump diffusions:
\begin{equation*}
dS_1 (t)= S_1 (t)\Big[\mu_tdt+\sigma_tdB(t)+\int_{\mathbb{R}\setminus\left\{ 0 \right\}} \eta(t,z)\tilde{N} (dt,dz)\Big],
\end{equation*}
where $\mu_\cdot$ and $\sigma_\cdot$ are deterministic continuous functions on $[0,T]$, $\eta(t,z):[0,T]\times \mathbb{R}\mapsto \mathbb{R}$ is a deterministic continuous function. What's more, we assume that $\mu_t> \rho_t$ for all $t\in[0,T]$.

We denote the wealth process of some investor as $X(\cdot)$. Suppose $\theta_0(t)$ is the number of the bond and $\theta_1(t)$ is the number of the risky asset held at time $t$, respectively. Then
\begin{equation*}
X (t)= \theta_0(t)S_0 (t)+\theta_1(t)S_1 (t),\quad t\in[0,T].
\end{equation*}
Assuming the portfolio selection satisfies the self-financing property, and $x>0$ represents the initial wealth. Then
\begin{equation*}
X (t)= x +\int_0^t \theta_0(s)dS_0 (s)+\int_0^t \theta_1(s)dS_1 (s).
\end{equation*}
We denote
\begin{equation*}
v (t)= \theta_1(t) S_1(t),\quad t\in[0,T],
\end{equation*}
as the amount of the wealth invested in the risky asset. So the wealth process $X(\cdot)$ can be described as
\begin{equation}\label{wealth}
dX (t)= \big[\rho_tX(t)+(\mu_t-\rho_t)v(t)\big]dt +\sigma_tv(t)dB(t) +v(t)\int_{\mathbb{R}\setminus\left\{ 0 \right\} } \eta(t,z)\tilde{N} (dt,dz).
\end{equation}

\begin{mydef}
The amount of the wealth invested in the risky asset $v(\cdot)$ is called admissible control if $v(\cdot)$ satisfies $\mathbb{E}\int_0^T v^2(t)dt< \infty$. The set of all admissible controls is denoted by $\mathcal{U}[0,T]$.
\end{mydef}

The mean-variance portfolio selection of jump diffusions, is to find an optimal control $\hat{v}(\cdot)$ to achieve the goal:
\begin{equation}\label{MF portfolio selection}
\left\{ J_1(\hat{v}(\cdot )),J_2(\hat{v}(\cdot )) \right\} = \min\limits_{v(\cdot)\in\,\mathcal{U}[0,T]} \left\{ J_1(v(\cdot )),J_2(v(\cdot )) \right\} \equiv \min\limits_{v(\cdot)\in\,\mathcal{U}[0,T]} \left\{ -\mathbb{E}X(T),\text{Var}X(T) \right \}.
\end{equation}

\begin{mydef}
An admissible control $\hat{v}(\cdot)$ is called efficient when there exists no admissible $v(\cdot)$ such that
\begin{equation}\label{efficient}
J_1(v(\cdot ))\le J_1(\hat{v}(\cdot)),\ J_2(v(\cdot))\le J_2(\hat{v}(\cdot)),
\end{equation}
and at least one of the inequalities holds strictly.
\end{mydef}

According to standard multi-objective optimization theory, multi-objective problems can be transformed into single-objective optimal control problems:
\begin{equation}\label{single objective}
\min\limits_{v(\cdot)\in\,\mathcal{U}[0,T]} \left\{ -\mathbb{E}X(T),\text{Var}X(T) \right\} \Leftrightarrow \min\limits_{v(\cdot)\in\,\mathcal{U}[0,T]} \left\{ -\mathbb{E}X(T)+\mu \text{Var}X(T) \right\},
\end{equation}
for some Lagrange multiplier $\mu>0$.
This problem remains difficult to deal with because of the $(\mathbb{E}[X(T)])^{2}$ term in the cost functional. Thanks to the work of \cite{ZL00}, this problem can be embedded into a stochastic LQ problem. We denote this problem by $P(\mu)$, and define
\begin{equation*}
\Pi_{P(\mu)} =\left\{ v\mid v(\cdot)\ \mbox{is an optimal control of } P(\mu) \right\}.
\end{equation*}
At the same time we consider the problem
\begin{equation}
\min\limits_{v(\cdot)\in\,\mathcal{U}[0,T]} \left\{ \mathbb{E}\left[\mu X^2(T)-\lambda X(T)\right] \right\} ,
\end{equation}
and define
\begin{equation*}
\Pi_{A(\mu,\lambda)} =\left\{ v\mid v(\cdot)\ \mbox{is an optimal control of } A(\mu,\lambda) \right\}  ,
\end{equation*}
where $-\infty< \lambda<+\infty$.
From the following result by Zhou and Li \cite{ZL00}, we know the relationship between problem $P(\mu)$ and $A(\mu,\lambda)$.
\begin{mypro}
{\bf (Embedding theorem)} For any $\mu>0$, we have
\begin{equation*}
 \Pi_{P(\mu)} \subseteq \Pi_{A(\mu,\lambda)},
\end{equation*}
Moreover, if $ \hat{v}(\cdot) \in \Pi_{P(\mu)}$, then $ \hat{v}(\cdot) \in  \Pi_{A(\mu,\lambda)}$ with $\hat{\lambda} = 1+2\mu \mathbb{E}\hat{X}(T)$, where $\hat{X}(\cdot)$ is the corresponding optimal wealth trajectory.
\end{mypro}
We need to point out that the embedding theorem does not depend on the form of state equation \eqref{wealth}.

To simplify the problem further, we apply the following variable substitution:
\begin{equation}\label{x and y}
\beta= \frac{\lambda }{2\mu} ,\ y(t)=\sqrt{\mu} (X(t)-\beta),\ u(t)=\sqrt{\mu}v(t).
\end{equation}
Then \eqref{MF portfolio selection} becomes
\begin{equation}\label{cost functional}
\min\limits_{u(\cdot)\in\,\mathcal{U}[0,T]} \left \{ \mathbb{E}\Big[\frac{1}{2}y^{2}(T) \Big] \right \} ,
\end{equation}
and \eqref{wealth} becomes
\begin{equation}\label{state SDE}
\left\{
\begin{aligned}
dy(t)&=\left[y(t)\rho_t +u(t)(\mu_t-\rho_t)+\sqrt{\mu} \beta \rho_t\right]dt+ u(t)\sigma_tdB(t)+u(t)\int_{\mathbb{R}\setminus\left\{ 0 \right\} }\eta(t,z)\tilde{N}(dt,dz),\\
y(0)&=y\equiv \sqrt{\mu}(x-\beta).
\end{aligned}
\right.
\end{equation}
Then we can solve this problem by MP and DPP approaches, respectively.

Before the end of this section, we introduce the following {\it generalized It\^{o}'s formula} (see, for example, \O ksendal and Sulem \cite{OS05}). Let $C_p^{1,2}([0,T]\times\mathbb{R})$ denote the set of functions $v:[0,T]\times\mathbb{R}\rightarrow\mathbb{R}$, satisfying $v(t,\cdot)$ is continuously differentiable on $[0,T]$, $v(\cdot,y )$ is second-order continuously differentiable on $\mathbb{R}$, and for constant $c>0$ and $k=1,2,\cdots$, $v(t,y)\le c(1+\left | y \right | ^k)$.
\begin{mylem}\label{generalized Ito formula}
{\bf (Generalized It\^{o}'s formula)} Suppose $dx(t)=b(t)dt+\sigma(t)dB(t)+\int_{\mathbb{R}\setminus\left\{ 0 \right\} }c(t,z)\\\tilde{N}(dt,dz)$, $\omega(\cdot,\cdot)\in C_p^{1,2}([0,T]\times\mathbb{R})$, then we have:
\begin{align}
d\omega(t,x(t))&=\left[\frac{\partial \omega}{\partial t} (t,x(t))+\frac{\partial \omega}{\partial x} (t,x(t))b(t)+\frac{1}{2} \frac{\partial^2 \omega}{\partial x^{2}}(t,x(t))\sigma^{2}(t)\right]dt\nonumber\\
&\quad +\frac{\partial \omega}{\partial x} (t,x(t))\sigma(t)dB(t)+\int_{\mathbb{R}\setminus\left\{ 0 \right\} }\big[\omega\big(t,x(t-)+c(t,z)\big)-\omega(t,x(t-))\big]\tilde{N}(dt,dz)\nonumber\\
&\quad +\int_{\mathbb{R}\setminus\left \{ 0 \right \} }\big[\omega\big(t,x(t)+c(t,z))-\omega(t,x(t)\big)-c(t,z)\frac{\partial \omega}{\partial x}(t,x(t)) \big]\lambda(dz)dt.\nonumber
\end{align}
\end{mylem}

\section{Solving the problem by MP}

Now we use the MP approach in \cite{FOS04} to solve our problem. According to the state equation \eqref{state SDE} and the cost functional in \eqref{cost functional}, we can write the Hamiltonian function as:
\begin{equation}
H(t,y,u,p,q,r)=[y\rho_t +u(\mu_t-\rho_t)+\sqrt{\mu} \beta \rho_t]p+u\sigma_tq+u\int_{\mathbb{R}\setminus\left\{ 0 \right\} }\eta(t,z) r(t,z)\lambda(dz),
\end{equation}
and the adjoint equation is:
\begin{equation}\label{adjoint}
\begin{cases}
dp(t)=-\rho_tp(t)dt +q(t)dB(t)+\int_{\mathbb{R}\setminus\left\{ 0 \right\} }r(t,z)\tilde{N}(dt,dz), \\
p(T)=-y(T).
\end{cases}
\end{equation}
Suppose $p(\cdot)$ has the form
\begin{equation}\label{relation of p, y}
p(\cdot)=\phi(\cdot)y(\cdot)+\psi(\cdot),
\end{equation}
for some differential functions $\phi(\cdot), \psi(\cdot)$, with $\phi(T)=-1, \psi(T)=0$.
Applying generalized It\^{o}'s formula (Lemma 2.1) to \eqref{relation of p, y}, we get
\begin{align}\label{dp}
dp(t)&=\left[\phi_ty(t)\rho_t +\phi_tu(t)(\mu_t-\rho_t)+\phi_t\sqrt{\mu} \beta \rho_t+y(t)\phi_t'+\psi_t'\right]dt\nonumber\\
&\quad + u(t)\sigma_{t}dB(t)+u(t)\int_{\mathbb{R}\setminus\left\{ 0 \right\} }\eta(t,z)\tilde{N}(dt,dz).
\end{align}
The following expressions are obtained by comparing the backward equation in \eqref{adjoint} and \eqref{dp}:
\begin{equation}\label{compare}
\begin{cases}
  \phi_ty(t)\rho_t +\phi_tu(t)(\mu_t-\rho_t)+\phi_t\sqrt{\mu} \beta \rho_t+y(t)\phi_t'+\psi_t'=-\rho_t(\phi_ty(t)+\psi_t),\\
  q(t)=\phi_t\sigma_tu(t),\\
  r(t,z)=\phi_tu(t)\eta(t,z),\quad \text{a.e.}t\in[0,T],\quad \mathbb{P}\text{-a.s.}
\end{cases}
\end{equation}
Substituting the above $q(\cdot)$, $r(\cdot,z)$ back into the Hamiltonian, we have
\begin{align}
&H(t,\hat{y}(t),\hat{u}(t), p(t),q(t),r(t))\nonumber\\
&=\rho_t\hat{y}(t)p(t)+\sqrt{\mu}\beta\rho_tp(t)+\hat{u}(t)\bigg[(\mu_{t}-\rho_{t})p(t)+\sigma_{t}q(t)+\int_{\mathbb{R}\setminus\left\{ 0 \right\} }\eta(t,z)r(t,z)\lambda(dz)\bigg].
\end{align}
The partial derivative of $H$ with respect to $u$ at $\hat{u}(\cdot)$ is equal to 0 achieves:
\begin{equation}\label{partial derivative of H}
(\mu_t-\rho_t)p(t)+\sigma_tq(t)+\int_{\mathbb{R}\setminus\left\{ 0 \right\} }\eta(t,z)r(t,z)\lambda(dz)=0,\quad \text{a.e.}t\in[0,T],\quad \mathbb{P}\text{-a.s.}
\end{equation}
Substituting \eqref{relation of p, y} and $q(\cdot), r(\cdot,\cdot)$ in \eqref{compare} into \eqref{partial derivative of H}, we have
\begin{equation}\label{optimal control of MP}
\hat{u}(t)=\frac{(\rho_t-\mu_t)(\phi_t\hat{y}(t)+\psi_t)}{\phi_t\Lambda_t},\quad \text{a.e.}t\in[0,T],\quad \mathbb{P}\text{-a.s.},
\end{equation}
where
\begin{equation}
\Lambda_{t}: =\sigma^{2}_{t}+\int_{\mathbb{R}\setminus\left \{ 0 \right \} }\eta^{2}(t,z)\lambda(dz).\nonumber
\end{equation}
What's more, from the first equality in \eqref{compare}, we can find another expression for $\hat{u}(\cdot)$:
\begin{equation}\label{optimal control of MP-another expression}
\hat{u}(t)=\frac{(\phi_t\rho_t+\phi'_t)\hat{y}(t)+\rho_t(\phi_t\hat{y}(t) +\psi_t)+\psi'_t+\sqrt{\mu}\beta\phi_t\rho_t}{\phi_t(\rho_t-\mu_t)},\quad \text{a.e.}t\in[0,T],\quad \mathbb{P}\text{-a.s.}
\end{equation}
Comparing \eqref{optimal control of MP} and \eqref{optimal control of MP-another expression}, we get:
\begin{equation*}
(\rho_t -\mu_t)^2(\phi_t\hat{y}(t)+\psi_t)=[(\phi_t\rho_t+\phi'_t)\hat{y}(t)+\rho_t(\phi_t\hat{y}(t)+\psi_t+\sqrt{\mu}\beta\phi_t)+\psi'_t]\Lambda_t,\quad t\in[0,T].
\end{equation*}
Comparing the coefficients of $\hat{y}(t)$ and the constant terms, we obtain two {\it ordinary differential equations} (ODEs):
\begin{equation}\label{phi}
\begin{cases}
(\rho_t-\mu_t)^2\phi_t-[2\rho_t\phi_t+\phi'_t]\Lambda_t=0,\\
\phi(T)=-1,
\end{cases}
\end{equation}
\begin{equation}\label{psi}
\begin{cases}
(\rho_t-\mu_t)^2\psi_t-[\rho_t\psi_t+\psi'_t+\phi_t\sqrt{\mu}\beta\rho_t]\Lambda_t=0,\\
\psi(T)=0.
\end{cases}
\end{equation}
Moreover if we define
\begin{equation*}
\theta(t):=\frac{(\mu_t-\rho_t)^2}{\Lambda_t},\quad t\in[0,T],
\end{equation*}
then \eqref{phi} and \eqref{psi} become:
\begin{equation}
\begin{cases}
\phi'(t)+(\theta(t)-2\rho_t)\phi(t) =0,\\
\phi(T)=1,
\end{cases}
\end{equation}
\begin{equation}
\begin{cases}
\psi'(t)-(\theta(t)-\rho_t)\psi(t)+\sqrt{\mu}\beta\rho_t\phi(t)=0,\\
\psi(T)=0,
\end{cases}
\end{equation}
which have explicit solutions:
\begin{equation}\label{phi solution}
\phi_t=-e^{\int_t^T(\theta(s)-2\rho_s)ds},\quad t\in[0,T].
\end{equation}
\begin{equation}\label{psi solution}
\psi_t=-\sqrt{\mu}\beta e^{-\int_t^T (\rho_s-\theta(s))ds}\left(e^{-\int_t^T \rho_sds}-1\right),\quad t\in[0,T].
\end{equation}
Substituting \eqref{phi solution} and \eqref{psi solution} into \eqref{optimal control of MP}, we get the following state/wealth feedback form:
\begin{equation}\label{optimal control of MP-state}
\hat{u}(t)=-\left[\hat{y}(t)+\sqrt{\mu}\beta\left(1-e^{-\int_t^T \rho_sds}\right)\right]\frac{\mu_t-\rho_t}{\Lambda_t},\quad \text{a.e.}t\in[0,T],\quad \mathbb{P}\text{-a.s.}
\end{equation}
We summarize the above analysis into the following theorem.
\begin{mythm}
By the MP approach, the state/wealth feedback form's optimal control $\hat{u}(\cdot)$ for our mean-variance portfolio selection of jump diffusions problem \eqref{MF portfolio selection} is given by \eqref{optimal control of MP-state}.
\end{mythm}

Efficient frontier is the optimal portfolio of risky assets that the market can offer. So it is very important for our mean-variance portfolio selection of jump diffusions problem. Now we try to find out the efficient frontier for our problem.  First by \eqref{x and y} and \eqref{optimal control of MP-state}, we have
\begin{equation}\label{optimal control of MP-state x}
\hat{v}(t)=\beta\left[e^{-\int_t^T\rho_sds }-\hat{X}(t)\right]\frac{\mu_t-\rho_t}{\Lambda_t},\quad \text{a.e.}t\in[0,T],\quad \mathbb{P}\text{-a.s.}
\end{equation}
Substituting \eqref{optimal control of MP-state x} into \eqref{wealth}, we have
\begin{equation}\label{optimal wealth}
\left\{
\begin{aligned}
d\hat{X}(t)&=\left[\hat{X}(t)\rho_t+\beta e^{-\int_t^T\rho_sds }\theta(t)-\hat{X}(t)\theta(t)\right]dt\\
     &\quad +\left[\beta e^{-\int_t^T\rho_sds}-\hat{X}(t)\right]\frac{(\mu_t-\rho_t)\sigma_t}{\Lambda_t}dB(t)\\
     &\quad +\int_{\mathbb{R}\setminus\left\{ 0 \right\} }\left[\beta e^{-\int_t^T\rho_sds}-\hat{X}(t)\right]\frac{(\mu_t-\rho_t)\eta(t,z)}{\Lambda_t}\tilde{N}(dt,dz) ,\\
 \hat{X}(0)&=x>0.
\end{aligned}
\right.
\end{equation}
Taking expectations to both sides, we get
\begin{equation}\label{expectation of X(t)}
\begin{cases}
d\mathbb{E}\hat{X}(t)=\left[(\rho_t-\theta(t))\mathbb{E}\hat{X}(t)+\beta e^{-\int_t^T\rho_sds}\theta(t)\right]dt ,\\
\hat{X}(0)=x.
\end{cases}
\end{equation}
The solution of \eqref{expectation of X(t)} is:
\begin{equation}\label{X(t) solution}
\mathbb{E}\hat{X}(t)=xe^{\int_0^t(\rho_s-\theta(s))ds}+\beta e^{-\int_t^T\rho_sds}\left(1-e^{-\int_0^t\theta(s)ds}\right).
\end{equation}
Applying Lemma 2.1 to $\hat{X}^2(\cdot)$ and take expectations, we get
\begin{equation}\label{dEX^2(t)}
\begin{cases}
d\mathbb{E}\hat{X}^2(t)=\left[(2\rho_{t}-\theta(t))\mathbb{E}\hat{X}^2(t)+\beta^{2}e^{-\int_{t}^{T}\rho_{s}ds }\theta(t)\right]dt,\\
\mathbb{E}\hat{X}^2(0)=x^2.
\end{cases}
\end{equation}
The solution of \eqref{dEX^2(t)} is
\begin{equation}\label{EX^2(t) solution}
\mathbb{E}\hat{X}^2(t)=x^2e^{\int_0^T(2\rho_s-\theta(s))ds}+\beta^2e^{-2\int_t^T\rho_sds}\left(1-e^{-\int_0^T\theta(s)ds}\right).
\end{equation}
In order to find out the relationship between $\mathbb{E}\hat{X}(T)$ and $\text{Var}\hat{X}(T)$, from \eqref{X(t) solution}, we get
\begin{equation}\label{expectation of X(t)-another form}
\beta\left(1-e^{-\int_0^T\theta(t)dt}\right)=\mathbb{E}\hat{X}(T)-xe^{\int_0^T(\rho_t-\theta(t))dt}.
\end{equation}
From \eqref{X(t) solution}, \eqref{EX^2(t) solution} and \eqref{expectation of X(t)-another form}, we get
\begin{align}\label{efficient frontier}
\text{Var}X(T)&=\mathbb{E}\hat{X}^2(T)-\big(\mathbb{E}\hat{X}(T)\big)^2 \nonumber\\
&=\frac{1}{e^{\int_0^T\theta(t)dt}-1} \left(\mathbb{E}\hat{X}(T)-xe^{\int_0^T \rho_tdt}\right)^2.
\end{align}
We summarize the above analysis into the following result.
\begin{mythm}
The efficient frontier of our mean-variance portfolio selection of jump diffusions problem \eqref{MF portfolio selection} is given by \eqref{efficient frontier}.
\end{mythm}

\section{Solving the problem by DPP}

In this section, we first establish the DPP and corresponding HJB equation of jump diffusions, then we use Theorem 4.1 to solve our problem in section 2. Though the result can be referred to Chapter 5 in \O ksendal and Sulem \cite{OS05}, we present here for the readers' convenience. The controlled state equation we considered is the following linear SDE of jump diffusions:
\begin{equation}\label{linear SDE of jump diffusions}
\begin{cases}
dy(t)=[A_ty(t)+B_tu(t)+C_t]dt+u(t)\sigma_{t}dB(t)+u(t)\int_{\mathbb{R}\setminus\left\{ 0 \right\} }\eta(t,z)\tilde{N}(dt,dz),
\\ y(0)=y,
\end{cases}
\end{equation}
where $A_\cdot,\ B_\cdot,\ C_\cdot,\ \sigma_\cdot$ are deterministic continuous function on $[0,T]$, $\eta(t,z):[0,T]\times \mathbb{R}\mapsto \mathbb{R} $ is a deterministic continuous function and $\int_0^T C^2_tdt< \infty$.
And we want to minimize the cost functional:
\begin{equation}\label{polynomial cost functional}
\mathbb{E}\left\{ \int_0^TL(s,y(s),u(s))+\Psi(T,y(T)) \right\},
\end{equation}
where for some constants $c>0$ and $k=1,2,\cdots$, uniformly in $t$,
\begin{equation*}
L(t,y,u)\le c(1+|y|^k+|u|^k),\quad \Psi(y)\le c(1+|y|^k).
\end{equation*}

We denote
\begin{equation}\label{polynomial cost functional-conditional}
J(t,y;u(\cdot)):=\mathbb{E}_{t,y}\left\{ \int_t^TL(s,y(s),u(s))+\Psi(y(T)) \right\},
\end{equation}
where $\mathbb{E}_{t,y}$ means the conditional expectation is taken when the initial condition is $y(t)=y$, $0\le t<T$. And we define the value function
\begin{equation}\label{value function}
V(t,y):=\min\limits_{u(\cdot)\in\,\mathcal{U}[0,T]} J(t,y;u(\cdot)),\quad (t,y)\in[0,T]\times\mathbb{R}.
\end{equation}

The following theorem gives the DPP and corresponding HJB equation of jump diffusions.
\begin{mythm}\label{DPP and HJB}
Let $\omega \in C_p^{1,2}([0,T]\times\mathbb{R})$, and satisfies:
\begin{equation}\label{HJB equation}
\begin{cases}
\min\limits_u \left\{ A^u\omega (t,y)+L(t,y,u) \right\} =0,\quad (t,y)\in[0,T]\times\mathbb{R},\\
\omega(T,y)=\Psi(y),\quad y\in\mathbb{R},
\end{cases}
\end{equation}
where
\begin{align}\label{integral-differential operator}
A^u\omega(t,y)&:=\frac{\partial \omega}{\partial t} (t,y)+\frac{\partial \omega}{\partial y}(t,y)[A_ty+B_tu+C_t]+\frac{1}{2} \frac{\partial\omega^2}{\partial y^2}(t,y)u^2\sigma^2_t\nonumber\\
&\quad +\int_{\mathbb{R}\setminus\left\{ 0 \right\} }\left[\omega(t,y+\eta(t,z))-\omega(t,y)-\eta(t,z)\frac{\partial \omega}{\partial y} (t,y)\right] \lambda(dz).
\end{align}
If there exists an admissible control $u^*(\cdot) $ such that
$$
u^*(\cdot) \in \argmin\limits_{u(\cdot)\in\,\mathcal{U}[0,T]} \left\{ A^u\omega (t,y(t))+L(t,y(t),u(t)) \right\},
$$
where $y^*(\cdot)$ is the solution to \eqref{linear SDE of jump diffusions} when $u(\cdot)=u^*(\cdot)$. Then $\omega(t,y)=V(t,y)$ and $u^*(\cdot)$ is an optimal control of the problem \eqref{value function}.
\end{mythm}

\begin{proof}
By Lemma \ref{generalized Ito formula}, we get
\begin{align*}
d\omega(t,y(t))&=\left[\frac{\partial \omega}{\partial t} (t,y(t))+\frac{\partial \omega}{\partial y} (t,y(t))[A_ty(t)+B_tu(t)+C_t]
 +\frac{1}{2} \frac{\partial^2 \omega}{\partial y^{2}}(t,y(t))u^2(t)\sigma^2_t\right]dt\\
&\quad +\frac{\partial \omega}{\partial y} (t,y(t))u(t)\sigma_tdB(t)
 +\int_{\mathbb{R}\setminus\left\{ 0 \right\} }\left[\omega(t,y(t-)+\eta(t,z))-\omega(t,y(t-))\right]\tilde{N}(dt,dz)\\
&\quad +\int_{\mathbb{R}\setminus\left\{ 0 \right\} }\left[\omega(t,y(t)+\eta(t,z))-\omega(t,y(t))-\eta(t,z)\frac{\partial \omega}{\partial y}(t,y(t))\right]\lambda(dz)dt.
\end{align*}
Thus by \eqref{integral-differential operator}, we have
\begin{align}\label{omega(t,y(t))}
\omega(t,y(t))&=\omega(T,y(T))-\int_t^TA^u\omega(s,y(s))ds -\int_{t}^{T}\frac{\partial \omega}{\partial y} (s,y(s))u(s)\sigma_sdB(s)\nonumber\\
&\quad -\int_t^T\int_{\mathbb{R}\setminus\left\{ 0 \right\} }[\omega(t,y(t-)+\eta(t,z))-\omega(t,y(t-))]\tilde{N}(dt,dz).
\end{align}
Taking $\mathbb{E}_{t,y}[\cdot]$ on both sides of \eqref{omega(t,y(t))}, we get
\begin{equation}
\mathbb{E}_{t,y}[\omega(T,y(T))]-\omega(t,y(t))=\mathbb{E}_{t,y}\left[\int_t^T A^u\omega(s,y(s))ds\right].
\end{equation}
From \eqref{HJB equation}, we know that for any $u \in \mathbb{R} $, we have
\begin{equation}
A^{u}\omega (t,y(t))+L(t,y(t),u)\ge 0.
\end{equation}
Thus we know that
\begin{align}
J(t,y;u)&=\mathbb{E}_{t,y} \left\{ \int_t^TL(s,y(s),u(s))+\Psi(y(T)) \right\},\nonumber\\
&\ge \mathbb{E}_{t,y} \left\{ -\int_t^TA^u\omega(s,y(s))ds+\omega(t,y(t))+ \int_t^TA^u\omega (s,y(s))ds \right\}=\omega(t,y(t)).
\end{align}
So
\begin{equation}\label{HJB inequality}
\min_{u}J(t,y;u)\ge \omega(t,y(t)).
\end{equation}
When $u=u^*(\cdot)$, \eqref{HJB inequality} becomes equality:
\begin{equation*}
V(t,y)=J(t,y^*(\cdot);u^*(\cdot))=\omega(t,y),
\end{equation*}
which means $u^*(\cdot)$ is an optimal control of the problem \eqref{value function}. So we have proved the result.
\end{proof}

Next, we use Theorem \ref{DPP and HJB} to our mean-variance portfolio selection of jump diffusions problem \eqref{MF portfolio selection}. Set $V(t,y)=\min\limits_u\mathbb{E}_{t,y}\left[\frac{1}{2}y^2(T)\right]$. Similar to \eqref{integral-differential operator}, we define an operator
\begin{align*}
A^uV(t,y)&:=\frac{\partial V}{\partial t} (t,y)+\frac{\partial V}{\partial y}(t,y)[y\rho_t+u(\mu_t-\rho_t)+\sqrt{\mu}\beta\rho_t]+\frac{1}{2} \frac{\partial V^2}{\partial y^2}(t,y)u^2\sigma^2_t\\
&\quad +\int_{\mathbb{R}\setminus\left\{ 0 \right\} }\left[V(t,y+\eta(t,z))-V(t,y)-\eta(t,z)\frac{\partial V}{\partial y} (t,y)\right] \lambda(dz).
\end{align*}
Then the HJB equation \eqref{HJB equation} now writes
\begin{equation}\label{HJB equation of mean-variance problem}
0=\min_{u}\left \{ A^{u}V(t,y) \right \},\quad (t,y)\in[0,T]\times\mathbb{R}.
\end{equation}
In order to find the solution of \eqref{HJB equation of mean-variance problem}, we set
\begin{equation}\label{value function of P, Q, R}
V(t,y)=\frac{1}{2}P(t)y^{2}+Q(t)y+R(t),\quad (t,y)\in[0,T]\times\mathbb{R},
\end{equation}
where $P(\cdot), Q(\cdot), R(\cdot)$ are differential functions with $P(T)=1, Q(T)=R(T)=0$. Substituting \eqref{value function of P, Q, R} into \eqref{HJB equation of mean-variance problem}, we have
\begin{align}\label{HJB equation of mean-variance problem by V(t,y)}
0&=\frac{1}{2}P'(t)y^2+Q'(t)y+R'(t) +\min_u\bigg\{ [y\rho_t+u(\mu_t-\rho_t)+\sqrt{\mu}\beta\rho_t](P(t)y+Q(t))\nonumber\\
&\quad + \frac{1}{2}P(t)u^2\sigma_t^2+ \frac{1}{2}P(t)u^2\int _{\mathbb{R}\setminus\left\{ 0 \right\} }\eta^2(t,z)\lambda(dz)\bigg\}.
\end{align}
Completing square for \eqref{HJB equation of mean-variance problem by V(t,y)}, we get
\begin{align}\label{completing square}
0&=\frac{1}{2}P(t)\min_u\left\{ \Lambda_t \left[u+\left(y+\frac{Q(t)}{P(t)}\right)\frac{D_t}{\Lambda_t}\right]^{2}  \right \}+\frac{1}{2} [P'(t)-(\theta(t)-2\rho_t)P(t)]y^2 \nonumber\\
&\quad +[Q'(t)-(\theta(t)-\rho_t)Q(t)+\sqrt{\mu}\beta\rho_tP(t)]y+R'(t)+\left[\sqrt{\mu}\beta\rho_t-\frac{1}{2}\theta(t)\frac{Q(t)}{P(t)}\right]Q(t),
\end{align}
where
\begin{equation*}
D_t:=\mu_t-\rho_t,\quad \Lambda_t:=\sigma^2_t+\int_{\mathbb{R}\setminus\left\{ 0 \right\} }\eta^2(t,z)\lambda(dz),\quad \theta(t):=\frac{D^2_t}{\Lambda_t}.
\end{equation*}
From \eqref{completing square}, we know that if and only if the following three ODEs
\begin{equation}\label{ODE P}
\begin{cases}
P'(t)-(\theta(t)-2\rho_t)P(t) =0,\\
P(T)=1,
\end{cases}
\end{equation}
\begin{equation}\label{ODE Q}
\begin{cases}
Q'(t)-(\theta(t)-\rho_t)Q(t)+\sqrt{\mu}\beta\rho_tP(t)=0,\\
Q(T)=0,
\end{cases}
\end{equation}
\begin{equation}\label{ODE R}
\begin{cases}
R'(t)+\left[\sqrt{\mu}\beta\rho_t-\frac{1}{2}\theta(t)\frac{Q(t)}{P(t)}\right]Q(t)=0,\\
R(T)=0,
\end{cases}
\end{equation}
admit solutions, problem \eqref{MF portfolio selection} has a unique optimal control
\begin{equation}\label{optimal control of DPP}
u^*(t)=-\left(y^*(t)+\frac{Q(t)}{P(t)}\right)\frac{D_t}{\Lambda_t},\quad \text{a.e.}t\in[0,T],
\end{equation}
where $y^*(\cdot)$ is the solution to \eqref{state SDE} corresponding to $u^*(\cdot)$.

Notice that \eqref{ODE P}, \eqref{ODE Q}, \eqref{ODE R} are all first-order ODEs, we can find out their unique solutions:
\begin{equation}\label{P}
P(t)=e^{\int_t^T(\theta(s)-2\rho_s)ds},
\end{equation}
\begin{equation}\label{Q}
Q(t)=\sqrt{\mu}\beta e^{-\int_t^T (\rho_s-\theta(s))ds}\left(e^{-\int_t^T \rho_sds}-1\right).
\end{equation}
Substituting \eqref{P} and \eqref{Q} into \eqref{optimal control of DPP}, we get
\begin{equation}\label{optimal control of DPP-state}
\hat{u}(t)=-\left[y^*(t)+\sqrt{\mu}\beta(1-e^{-\int_t^T \rho_sds})\right]\frac{\mu_t-\rho_t}{\Lambda_t},\quad \text{a.e.}t\in[0,T].
\end{equation}
It stands to reason that this expression is the same as \eqref{optimal control of MP-state}.

We summarize the above process into the following result.
\begin{mythm}
By the DPP approach, the state/wealth feedback form's optimal control $\hat{u}(\cdot)$ for our mean-variance portfolio selection of jump diffusions problem \eqref{MF portfolio selection} is given by \eqref{optimal control of DPP-state}.
\end{mythm}

Moreover, comparing \eqref{phi solution} and \eqref{P}, we discover
\begin{equation}\label{phi and P}
\phi_{t}=-P(t),
\end{equation}
and comparing \eqref{psi solution} and \eqref{Q}, we find
\begin{equation}\label{psi and Q}
\psi_{t}=-Q(t),
\end{equation}
which inspire us to seek out the relationship between MP and DPP.

\section{Relationship between MP and DPP}

After the calculations in section 3 and section 4, we can verify the following result directly, showing the relationship between MP and DPP for our mean-variance portfolio selection of jump diffusions problem.
\begin{mythm}
{\bf (Relationship between MP and DPP)}\quad For problem \eqref{MF portfolio selection}, let $u^*(\cdot)$ be an optimal control and $y^*(\cdot)$ is the corresponding optimal state/wealth trajectory satisfying \eqref{state SDE}, then the following results hold:
\begin{equation}\label{relation between MP and DPP}
\begin{cases}
p(t)=-\frac{\partial V}{\partial y}(t,y^*(t)),\\
q(t)=-\frac{\partial^2 V}{\partial y^2}(t,y^*(t))\sigma_tu^*(t),\\
-r(t,z)=\frac{\partial V}{\partial y}\big(t,y^*(t)+u^*(t)\eta(t,z)\big)-\frac{\partial V}{\partial y}(t,y^*(t)),\quad \text{a.e.}t\in[0,T],\quad \mathbb{P}\text{-a.s.},
\end{cases}
\end{equation}
where $(p(\cdot),q(\cdot),r(\cdot,\cdot))$ satisfies the adjoint equation \eqref{adjoint}, and $V(\cdot,\cdot)$ is the value function.
\end{mythm}

\begin{proof} It is direct by \eqref{relation of p, y}, \eqref{compare}, \eqref{value function of P, Q, R}, \eqref{phi and P} and \eqref{psi and Q}. We omit the detail.
\end{proof}

\section{Conclusions}

In this paper, we have studied the maximum principle and dynamic programming principle for mean-variance portfolio selection of jump diffusions and their relationship. First, the optimal portfolio and efficient frontier of the problem are obtained using both methods. Furthermore, the relationship between these two methods is investigated. Specially, the connections between the adjoint processes and value function are given.

In the future, we will consider the problems with recursive utilities (see \cite{Shi14}, \cite{LWZ15}, \cite{SGZ18}, \cite{Li23}, \cite{WS24}).

\end{document}